\documentclass[11pt]{article}
\usepackage[utf8]{inputenc}
\usepackage{times}
\usepackage[ruled,vlined,commentsnumbered,titlenotnumbered]{algorithm2e}
\usepackage{url}
\usepackage{fullpage}
\usepackage{pslatex} 
\usepackage{mathdots} 
\usepackage{amsmath}
\usepackage{algorithmic}
\usepackage{amssymb}
\usepackage{amsthm}
\usepackage{color}
\usepackage{array}
\usepackage{xy}
\usepackage{setspace}
\usepackage{multicol}
\usepackage{hyperref}
\usepackage{xspace}
\usepackage{cite}
\usepackage[margin=1 in]{geometry}

\usepackage{thmtools}
\usepackage{thm-restate}

\newtheoremstyle{slanted}
{3pt}
{3pt}
{\slshape}
{}
{\bfseries}
{.}
{.5em}
{}
\theoremstyle{slanted}
\newtheorem{theorem}{Theorem}

\newtheorem{definition}[theorem]{Definition}

\usepackage[subtle, mathspacing=normal]{savetrees}

\usepackage{tikz}

\def\eps{\epsilon}
\def\Oreg{\mbox{O}}
\def\Otld{\mbox{\~{O}}}

\def\chr#1{{\tt{#1}}}
\def\str#1{{\tt{#1}}}

\def\pequ{\simeq}

\def\pjuxt{+}
\def\pjoin{\oplus}

\def\pathseq#1{\langle{#1}\rangle}

\def\WP{{\mathcal P}}
\def\CC{\mbox{\it cost}}
\def\DD{\mbox{\it dist}}
\def\LL{\mbox{\it len}}
\def\DTW{\mbox{\it DTW}}

\def\DTWtwo{\tilde{\mbox{\it DTW}}}

\def\GDTW{G_{\it D\kern-0.70ptT\kern-1.40ptW}}
\def\TDTW{T_{\it D\kern-0.70ptT\kern-1.40ptW}}

\def\myran#1{[#1]}
\def\myset#1{\{#1\}}
\def\myrle#1{\hat{#1}}
\def\mylen#1{|#1|}
\def\myrlen#1{\|#1\|}
\def\mydist{\mbox{$\delta$}}
\def\mynear{\mbox{\it snap}}
\def\letters{\mbox{$\Sigma$}}

\def\edge{\mbox{\textbf e}}
\def\vrtx{\mbox{\textbf v}}
\def\pairg#1{\langle{#1}\rangle}
\newcommand{\defn}[1]{\emph{\textbf{{#1}}}}
\newcommand{\poly}{\operatorname{poly}}
\newcommand{\polylog}{\operatorname{polylog}}
\renewcommand{\paragraph}[1]{\vspace{.5 cm} \noindent \textbf{#1} }
\def\BB{\mbox{\bf B}}

\def\mylr{\mbox{\it lr}}

\title{%
Approximating Dynamic Time Warping Distance\break Between Run-Length Encoded Strings
}

\author{Zoe Xi\footnote{\textit{zoexi@bu.edu}} \\ Boston University Academy \and William Kuszmaul\footnote{\textit{kuszmaul@mit.edu}} \\ Massachusetts Institute of Technology}

\date{}
\begin{document}

\maketitle


\begin%
{abstract}
Dynamic Time Warping (DTW) is a widely used similarity measure for
comparing strings that encode time series data, with applications to
areas including bioinformatics, signature verification, and speech recognition.  
The standard dynamic-programming algorithm for DTW takes $O(n^2)$ time,
and there are conditional lower bounds showing that no algorithm can 
do substantially better. 

In many applications, however, the strings $x$ and $y$ may contain
long runs of repeated letters, meaning that they can be compressed
using run-length encoding.  A natural question is whether the
DTW-distance between these compressed strings can be computed
efficiently in terms of the lengths $k$ and $\ell$ of the compressed
strings.  Recent work has shown how to achieve $O(k\ell^2 + \ell
k^2)$ time, leaving open the question of whether a near-quadratic
$\tilde{O}(k\ell)$-time algorithm might exist.

We show that, if a small approximation loss is permitted, then a
near-quadratic time algorithm is indeed possible: our algorithm
computes a $(1 + \epsilon)$-approximation for $\DTW(x, y)$ in
$\tilde{O}(k\ell / \epsilon^3)$ time, where $k$ and $\ell$ are the
number of runs in $x$ and $y$. Our algorithm allows for $\DTW$ to be
computed over any metric space $(\Sigma, \delta)$ in which distances
are $\Oreg(\log n)$-bit integers. Surprisingly, the algorithm also
works even if $\delta$ does not induce a metric space on $\Sigma$
(e.g., $\delta$ need not satisfy the triangle inequality).
\end{abstract}

\section{Introduction}
\label{section:introduction}
Dynamic Time Warping (DTW) distance is a well-known similarity measure
for comparing strings that represent time-series data. DTW distance
was first introduced by Vintsyuk in 1968~\cite{vintzyuk1968}, who
applied it to the problem of speech discrimination. In the decades
since, DTW has become one of the most widely used similarity
heuristics for comparing time series~\cite{DBLP:journals/pr/Liao05} in
applications such as bioinformatics, signature verification, and
speech
recognition~\cite{dtwapp1,dtwapp2,dtwapp3,dtwapp4,dtwapp5,dtwapp6}.

Consider any two strings $x$ and $y$, with characters taken from some
metric space $(\Sigma, \delta)$. For example, in many applications, we
have that $\Sigma = \mathbb{R}^{c}$ for some parameter $c$ and that
$\delta(a, b) = \| a-b\|_2$ computes $\ell_2$ distance.  Define a
\defn{time warp} of $x$ (and similarly of $y$) to be any string $x'$
that can be obtained by \defn{warping} the letters in $x$, where
warping a letter means replacing it with $\ge 1$ consecutive copies of
itself. The DTW-distance $\DTW(x, y)$ is defined to be
\[
\min_{|x'| = |y'|} \sum_{i = 1}^{|x'|} \delta(x'_i, y'_i),
\]
where $x'$ and $y'$ range over all time warps of $x$ and $y$. 

The most fundamental question concerning DTW is how to compute it
efficiently. Vintsyuk showed that, given strings $x$ and $y$ of length
$n$, it is possible to compute $\DTW(x, y)$ in $\Oreg(n^2)$
time~\cite{vintzyuk1968}. His algorithm, which was one of the earliest
uses of dynamic programming, continues to be taught in textbooks and
algorithms courses today.

For many decades, it was an open question whether any algorithm could
achieve a running time of $\Oreg(n^{2 - \Omega(1)})$. (Interestingly,
it is known that one \emph{can} shave small sub-polynomial factors off
of the running time~\cite{DBLP:conf/icalp/GoldS17}.)  A major
breakthrough occurred in 2015, when Abboud, Backurs, and Williams
\cite{DTWhard2} and Bringmann and K{\"u}nnemann \cite{DTWhard}
established conditional lower bounds prohibiting any strongly
subquadratic-time algorithm for DTW, unless the Strong Exponential
Time Hypothesis (SETH) fails.

This lower bound puts us in an interesting situation. On one hand, the
classic $\Oreg(n^2)$-time algorithm is often too slow for practical
applications. On the other hand, we have good reason to believe that
it is nearly optimal. This has led researchers to focus on forms of beyond-worst-case
analysis when studying the DTW problem.

An especially appealing
question~\cite{DBLP:journals/kais/SharabianiDHDKJ18,DBLP:journals/corr/abs-1903-03003}
is what happens if $x$ and $y$ both contain long runs of repeated
letters.  In this case, the strings can be compressed using run-length
encoding (RLE). For example, the string ``$\str{aaaaabbc}$'' has RLE
encoding ``($\chr{a}$, 5), ($\chr{b}$, 2), ($\chr{c}$, 1)''. If a
string $x$ has $k$ runs, then it is said to have an RLE representation
of length $k$.

It is known that, if $x$ and $y$ each contain $k$ runs, then $\DTW(x,
y)$ can be computed in $\Oreg(k^3)$
time~\cite{DBLP:journals/corr/abs-1903-03003}.\footnote{%
More generally, if $x$ contains $k$ runs and $y$ contains $\ell$ runs,
then the time becomes $\Oreg(k\ell^2 + \ell k^2)$.
}
It is still an open question whether it is possible to significantly
reduce this cubic running time, and in particular, whether a
near-quadratic time algorithm might be possible.

\paragraph
{This paper: A Near-Quadratic Approximation Algorithm.}  
We show that, if a small approximation loss is permitted, then a
near-quadratic time algorithm is indeed possible. Consider any two
run-length encoded strings $x \in \Sigma^{n}$ and $y \in \Sigma^{n}$,
where $x$ has $k$ runs and $y$ has $\ell$ runs. Let $\delta$ be an
arbitrary distance function $\delta: \Sigma \times \Sigma \rightarrow
[\poly(n)]$ mapping pairs of characters to $\Oreg(\log n)$-bit
nonnegative integers. (Perhaps surprisingly, our algorithms will not
require $\delta$ to satisfy the triangle inequality, or even to be
symmetric.)

Our main result is an algorithm that computes a $(1+\eps)$-approximation for $\DTW(x, y)$
in $\Otld(k\ell/\eps^3)$ time\footnote{Here we are using
  soft-O notation to mean that $\Otld(k\ell/\eps^3)$ is equivalent to
  $\Oreg((k\ell/\eps^3)\polylog(n))$.}. In the special case where
 $\Sigma$ is over Hamming space (i.e., $\delta(a, b)$ is either $0$ or $1$ for all $a, b \in
\Sigma$), the running time of our algorithm further improves to $\Otld(k\ell / \eps^2)$. 

Our algorithm takes a classical geometric interpretation of DTW in
terms of paths through a grid, and shows how to decompose each path in
such a way that its components can be efficiently approximated. This
allows for us to reduce the problem of approximating DTW-distance
between RLE strings to the problem of computing pairwise distance in a
small directed acyclic graph.

\paragraph{Other related work.}
In addition to work on run-length-encoded strings
\cite{DBLP:journals/kais/SharabianiDHDKJ18,DBLP:journals/corr/abs-1903-03003},
there has been a recent push to study other theoretical facets of the
DTW problem. This includes work on approximation algorithms
\cite{DBLP:conf/icalp/Kuszmaul19, DTWapprox1, DTWapprox2},
low-distance-regime algorithms \cite{DBLP:conf/icalp/Kuszmaul19},
communication complexity \cite{dtwcomm}, slightly-subquadratic
algorithms \cite{DBLP:conf/icalp/GoldS17}, reductions to other
similarity measures
\cite{DBLP:conf/icalp/Kuszmaul19,sakai2020reduction, sakai2020faster},
binary DTW \cite{kuszmaul2021binary, schaar2020faster}, etc.

All of these results (along with the results in this paper) can be
viewed as part of a larger effort to close the gap between what is
known about DTW and what is known about its closely related cousin
\defn{edit distance}, which measures the number of insertions,
deletions, and substitutions of characters needed to turn one string
$x$ into another string $y$. Like DTW, edit distance can be computed
in $\Oreg(n^2)$ time using dynamic programming \cite{vintzyuk1968,
  needleman1970general} (and can be computed in \emph{slightly}
subquadratic time using lookup-table techniques
\cite{masek1980faster}).  Also like DTW, edit distance has
conditional lower bounds \cite{DTWhard,
  DTWhard2,DBLP:conf/icalp/Kuszmaul19} prohibiting strongly
subquadratic time algorithms.

When it comes to beyond-worst-case analysis, however, edit distance
has yielded much stronger results than DTW: it is known how to compute
a constant-approximation for edit distance in strongly subquadratic
time \cite{edapprox1, edapprox2, edapprox3, edapprox4, edapprox5,
  edapprox6, edapprox7, edapprox8}; it is known how to compute the
edit distance between RLE strings in $\tilde{O}(k\ell)$ time
\cite{rle1, rle2, rle3, rle4, rle5, rle6, rle7, rle8, rle9}; and if
two strings $x$ and $y$ have small edit distance $k$, it is known how
to compute the edit distance in $\Oreg(|x|+|y|+ k^2)$ time
\cite{chakraborty2016streaming, landau1998incremental}.

Whether or not any of these results can be replicated for DTW remains
the central open question in modern theoretical work on DTW. There are
several reasons to believe that DTW computation should be more
challenging than edit distance. Whereas edit distance satisfies the
triangle inequality, DTW does not (for example, if we take $\Sigma =
\{0, 1\}$, then $\DTW(111110, 100000) = 0$, $\DTW(100000, 000000) =
1$, and $\DTW(111110, 000000) = 5$). This erratic behavior of DTW
seems to make it especially difficult to approximate. Additionally,
whereas almost all work on edit distance focuses on
insertion/deletion/substitution costs of $1$, work on DTW must
consider arbitrary cost functions $\delta$ for comparing
characters. Finally, it is known that the problem of computing edit
distance actually \emph{reduces} to that of computing DTW
\cite{DBLP:conf/icalp/Kuszmaul19}, indicating that the latter problem
is at least as hard (although, interestingly, this reduction
\emph{does not} apply in the run-length encoded setting).

Our paper represents the first evidence that an
$\tilde{O}(k\ell)$-time algorithm for DTW may be within reach. Such an
algorithm would finally unify edit distance and DTW in the
run-length-encoded setting.

\section{Technical Overview}\label{sec:technical}

This section gives a technical overview of how we approximate
$\DTW$-distance between run-length encoded strings. To simplify
exposition, we focus here only on the big ideas in the algorithm
design and we defer the detailed analysis to later sections.

Throughout the section, we consider two strings $x$ and $y$ of length
$n$ whose characters are taken from a set $\Sigma$ with a
symmetric distance function $\delta: \Sigma \times \Sigma \rightarrow
\mathbb{N} \cup \{0\}$. Our only assumption on $\delta$ is that
$\delta(a, b) \in \{0, 1, 2, \ldots, \poly(n)\}$ for all $a, b \in
\Sigma$. (We do not need the triangle inequality on $\delta$.) Let $k$
and $\ell$ be the number of runs in $x$ and $y$, respectively.  We
will describe a $(1 + \Oreg(\eps))$-approximate algorithm that takes
$\tilde{O}(k \ell / \poly(\eps))$ time.

\paragraph{How to think about $\DTW$.}
There are several mathematically equivalent ways (see, e.g.,
\cite{DBLP:conf/icalp/Kuszmaul19,
  DBLP:conf/icalp/GoldS17,DBLP:journals/corr/abs-1903-03003}) to
define the dynamic time warping distance between $x$ and $y$. In this
paper, we work with the geometric interpretation: consider an $n
\times n$ grid where cell $(i, j)$ has \defn{cost} $\delta(x_i, y_j)$;
consider the \defn{paths} through the grid that travel from $(1, 1)$
to $(n, n)$ via steps of the form $\langle 1, 0 \rangle$ (a horizontal
step (h-step) to the right), $\langle 0, 1\rangle$ (a vertical step
(v-step) up), and $\langle 1, 1\rangle$ (a diagonal step (d-step) to
the upper right); the \defn{cost} of such a path is the sum of the
costs of the cells that it encounters, and $\DTW(x, y)$ is defined to
be the smallest cost of any such path. For an example, see Figure
\ref{figure:DTWexample1}, which shows an optimal full path for computing
$\DTW(\str{aaabbbbddd}, \str{aabcdd}) = 1$, where the
$\delta$-function measures the distance between characters in the
alphabet.

Note that the $i$-th column of the grid corresponds to $x_i$ and the
$j$-th row of the grid corresponds to $y_j$. Thus, each run $x_{i_0},
\ldots, x_{i_1}$ in $x$ corresponds to a sequence of adjacent columns
$i_0, \ldots, i_1$ in the grid, and each run $y_{j_0}, \ldots,
y_{j_1}$ in $y$ corresponds to a sequence of adjacent rows $j_0,
\ldots, j_1$ in the grid.

If we want to design an algorithm that approximates $\DTW(x, y)$ in
$\tilde{O}(k \ell / \poly(\eps))$ time, then it is natural to think
about the grid as follows. We break the grid into \defn{blocks} by
drawing a vertical line between every pair of runs in $x$ and a
horizontal line between every pair of runs in $y$; and label the
blocks $\{\BB_{i, j}\}_{i\in[k], j\in[\ell]}$, where block $\BB_{i,
  j}$ corresponds horizontally to the $i$-th run in $x$ and vertically
to the $j$-th run in $y$.  All of the cells within a given block $B$
have the same cost, which we refer to as $\delta(B)$. We may also use
$\delta_{i,j}$ for $\delta(\BB_{i,j})$.  We refer to the first/last
row of each block as a lower/upper \defn{horizontal boundary} and to
the first/last column of each block as a left/right \defn{vertical
  boundary}.

Finally, it will be helpful to talk about sequences of blocks that are
adjacent horizontally or vertically. An \defn{h-block segment}
consists of a sequence of consecutive blocks lined up
horizontally. Formally, given $i_1\leq i_2$ in $[k]$ and $j$ in
$[\ell]$, we use $\BB_{[i_1,i_2],j}$ for the h-block segment
$\BB_{i_1,j},\BB_{i_1+1,j},\ldots,\BB_{i_2,j}$. Similarly, a
\defn{v-block segment} consists of a sequence of consecutive blocks
lined up vertically---we use $\BB_{i,[j_1,j_2]}$ for the v-block
segment $\BB_{i,j_1},\BB_{i,j_1+1},\ldots,\BB_{i,j_2}$. In the same
way that we can talk about the four boundaries of a block, we can talk
about the four boundaries of a given h-block or v-block segment.

Intuitively, since there are $\Oreg(k\ell)$ blocks, our goal is to design
an algorithm that runs in time roughly proportional to the number of
blocks.

\paragraph{How to think about the optimal path.}
Let $P$ be a minimum-cost path through the grid. We can decompose the
path into a sequence of disjoint \defn{components} $P_1, P_2, P_3,
\ldots$, where each component takes one of two forms:
\begin%
{enumerate}
\item
A \textbf{horizontal-to-vertical} (h-to-v) component connects a cell
on the lower boundary of some v-block segment to another cell on the
right boundary of the same v-block segment.
\item
A \textbf{vertical-to-horizontal} (v-to-h) component connects a cell
on the left boundary of some h-block segment to another cell on the
upper boundary of the same h-block segment.
\end{enumerate}
The components $P_1, P_2, P_3, \ldots$ are defined such that the end
cell of each $P_r$ connects to the the start cell of each $P_{r+1}$
via a single step (either horizontal, vertical, or diagonal).

We will now describe a series of simplifications that we can make to
$P$ while increasing its total cost by at most a $(1 +
\Oreg(\eps))$-factor. The simplifications are central to the design of
our algorithm.

\paragraph%
{Simplification 1: Rounding each component to start and end on ``snap
  points''.}
Let us call a grid cell $(i, j)$ an \defn{intersection point} if it
lies in the intersection of a horizontal boundary and a vertical
boundary. (Each block contains at most four intersection points.)  We
call a grid cell a \defn{snap point} if either it is an intersection
point, or it is of the form $(i+(1+\eps)^{t}, j)$ on the upper
boundary of a block $B$, or it is of the form $(i+1+(1+\eps)^{t}, j)$
on the lower boundary of a block $B$, or it is of the form $(i,
j+(1+\eps)^{t})$ on the right boundary of a block $B$, or it is of the
form $(i, j+1+(1+\eps)^{t})$ on the left boundary of a block $B$,
where $(i, j)$ is an intersection point of the block $B$ and $t$ is
nonnegative integer (since this is a technical overview, we ignore
floor and ceiling issues). For each boundary cell $p$ in the grid,
define $\mynear(p)$ to be the nearest snap point to the right of $p$,
if $p$ is on a horizontal boundary, and to be the nearest snap point
above $p$, if $p$ is on a vertical boundary. If $p$ is on both a
horizontal boundary and a vertical boundary, then $p$ is an
intersection point, so $\mynear(p) = p$.

How much would the cost of $P$ increase if we required each of its
components to start and end on snap points? Suppose, in particular,
that we replace each component $P_r$ with a component $P'_r$ whose
start point $p_r$ has been replaced with $\mynear(p_r)$ and whose end
point $q_r$ has been replaced with $\mynear(q_r)$.  It may be that
$\mynear(q_r)$ does not connect to $\mynear(p_{r + 1})$, meaning that
$P_r'$ and $P_{r + 1}'$ do not connect properly. If this happens,
however, then one can simply modify the starting-point of $P_{r + 1}'$
in order to connect it to $P_r'$ (and it turns out this only makes
$P_{r + 1}'$ cheaper).


Let $P'$ be the concatenation of $P'_1$, $P'_2$, $P'_3$, $\ldots$.  To
bound the cost of $P'$, we can argue that the cost of each $P'_r$ is
at most $(1+\eps)$ times that of $P_r$. To transform $P_r$ into
$P'_r$, the first step is to round the start point $p_r$ of $P_r$ to
$\mynear(p_r)$---one can readily see that this only decreases (or
leaves unchanged) the cost of $P_r$. The second step is to round the
end point $q_r$ of $P_r$ to $\mynear(q_r)$. For simplicity, assume
that $P_r$ is an h-to-v component that starts on the lower boundary of
some block $\BB_{i,j_1}$ and finishes on the right boundary of some
block $\BB_{i, j_2}$. Let $(u, v)$ be the lower-right intersection
point of $\BB_{i, j_2}$ and suppose that $P_r$ finishes in cell
$(u,v+s)$. Then $P_r$ incurs cost at least $(s+1)\cdot\delta_{i, j_2}$ in
block $\BB_{i, j_2}$. Moreover, the snap point $\mynear(q_r) =
\mynear(u,v+s)$ is guaranteed to be in the set $\{(u, v+s+t)\}_{t \in
  \{0, 1, \ldots, \eps\cdot{s}\}}$. Thus the cost of traveling from
$q_r$ to $\mynear(q_r)$ is at most $\eps\cdot{s}\cdot\delta_{i,
  j_2}$. So the cost of $P'_r$ is at most $(1+\eps)$ times that of
$P_r$.

By analyzing each component in this way, we can argue that
$\operatorname{cost}(P') \le (1+\eps)
\operatorname{cost}(P)$. Throughout the rest of the section, we will
assume that $P$ has been replaced with $P'$, meaning that each
component starts and ends with a snap point.

\paragraph{Simplification 2: Understanding the structure of each component.}
Next we observe that each individual component can be assumed to have
a relatively simple structure. For simplicity, let us focus on an
h-to-v component $P_r$ in a v-block segment $\BB_{i,[j_1,j_2]}$.  We
may assume without loss of generality that all of $P_r$'s h-steps
occur together on the lower boundary of some block; and that all of
$P_r$'s v-steps occur at the end of $P_r$. In other words, $P_r$ is of
the form $D_1\pjoin{H}\pjoin{D_2}\pjoin{U}$ where $\pjoin$ is for path
concatenation, $D_1$ consists of d-steps, $H$ consists of h-steps
(along a lower boundary), $D_2$ again consists of d-steps, and $U$
consists of v-steps (along a right boundary).\footnote{Note that the
  components $D_1$, $H$, $D_2$, and $U$ are each individually allowed
  to be length 0.} (See Figure~\ref{figure:h-to-v-path-approximation}
where the path $p_1q_1q_2p_2p_3$ in solid lines is such an example.)

Combined, these assumptions make it so that $P_r$ is fully determined
by four quantities: (1) $P_r$'s start point $p_r$, (2) the block
$\BB_{i, j}$ in which $H$ occurs, (3) the length of $H$, and (4) the
length of $U$.

Define $\overline{P_r}$ to be the prefix of $P_r$ that
terminates as soon as $U$ hits its first snap point.  (See
Figure~\ref{figure:h-to-v-path-approximation} where the path
$p_1q_1q_2p_2p'_2$ is such a prefix of the path $p_1q_1q_2p_2p_3$.)  We
will see later that $\overline{P_r}$ is, in some sense, the
``important'' part of $P_r$ to our algorithm. Observe that
$\overline{P_r}$ is fully determined by just three quantities: (1)
$P_r$'s start point $p_r$, (2) the block $\BB_{i, j}$ in which $H$
occurs, and (3) the length of $H$.

\paragraph{Simplification 3: Reducing the number of options for each component.} 
We will now argue that, if we fix the start point $p_r$, and we are
willing to tolerate a $(1 + \Oreg(\eps))$-factor approximation loss, then we only need to consider $\poly(\eps^{-1} \log n)$ options for $\overline{P_r}$. 

We begin by considering block $\BB_{i,j}$ in which $H$ occurs. Let us
define the sequence of blocks $B_0, B_1, B_2, \ldots$ so that $B_s =
\BB_{i, j + s}$ and define the sequence of costs $\delta_0,
\delta_1, \delta_2, \ldots$ so that $\delta_s = \delta_{i, j +
  s}$. We say that a block $B_s$ is \defn{extremal} if $(1 +
\eps)\delta_s \le \delta_{t}$ for all $t < s$. If we are willing to
tolerate a $(1 + \Oreg(\eps))$-factor increase in $\overline{P_r}$'s cost, then we
can assume without loss of generality that $H$ occurs in an extremal
block. On the other hand, there are only $\Oreg(\log_{1+\eps}(n))$
extremal blocks, so this means that we only need to consider
$\Oreg(\log_{1+\eps}(n)) \le \poly(\eps^{-1} \log n)$ options for the
starting point of $H$.

Next we consider the length of the horizontal sub-component $H$. If we
are willing to tolerate a $(1 + \Oreg(\eps))$-factor increase in $\overline{P_r}$'s
cost, then we can round $|H|$, the length of $H$, up to be a power of
$(1+\eps)$ (or to be whatever length brings us to the next vertical
boundary). Thus we only need to consider $\Oreg(\log_{1+\eps}(n)) \le
\poly(\eps^{-1} \log n)$ options for $|H|$.


Together, the block $\BB_{i, j}$ in which $H$ occurs and the length of $H$ fully determine $\overline{P_r}$. Thus, we have reached the following conclusion:
if the start point $p_r$ of the component $\overline{P_r}$ is known, then there
are only $\poly(\eps^{-1} \log n)$ options that we must consider for what $\overline{P_r}$ could look like. Moreover, although we have considered only
h-to-v components here, one can make a similar argument for v-to-h components. 

\paragraph{Approximating $\DTW$ in $\tilde{O}(k \ell / \poly(\eps))$ time.} 
We will now construct a weighted directed acyclic graph $G = \pairg{V, E}$
that has two special vertices $\vrtx_0$ and $\vrtx_{*}$ and that satisfies the
following properties:
\begin{itemize}
\item $G$ has a total of $\tilde{O}(k \ell / \poly(\eps))$ vertices/edges, and
\item the distance from $\vrtx_0$ to $\vrtx_{*}$ in $G$ is a $(1 + \Oreg(\eps))$-approximation for $\DTW(x, y)$. 
\end{itemize}
This reduces the problem of approximating $\DTW(x, y)$ to the problem
of computing a distance in a weighted directed acyclic graph. The
latter problem, of course, can be solved in linear time with dynamic
programming; thus the graph $G$ give us a $\tilde{O}(k \ell /
\poly(\eps))$-time $(1 + \Oreg(\eps))$-approximation algorithm for
$\DTW$.

We construct $G$ to capture the different ways in which path
components $P_r$ can connect together (assuming that the path
components take the simplified forms described above). As the vertices
$v \in V$ correspond to the snap points $p$ in the grid, we can use a
vertex to refer to its corresponding snap point and vice versa.  We
define $\vrtx_0$ to be the cell $(1, 1)$ in the grid and $\vrtx_{*}$
to be the cell $(n, n)$.  We add edges $E$ as follows:
\begin{itemize}
\item
We connect each snap point $p$ on a horizontal (resp. vertical)
boundary to the next snap point $q$ to its right (resp. above it).

\item
We connect each snap point $p$ on a right (resp. upper) boundary to
any snap points $q$ on the adjacent left (resp. lower) boundary that
can be reached from $p$ in a single step.

\item
Each snap point $p \in V$ has $\poly(\eps^{-1} \log n)$ out-edges
corresponding to the $\poly(\eps^{-1} \log n)$ options for what a
(truncated) component $\overline{P_r}$ starting at $p$ could look
like.\footnote{Note that $ G $ is not necessarily simple. If there are
  multiple ways that a component $\overline{P_r}$ could connect two
  vertices $p_1$ and $p_2$, then there will be multiple edges from
  $p_1$ to $p_2$.}
\end{itemize}
Note that, although we only add edges for \emph{truncated} path
components $\overline{P_r}$ (rather than full components $P_r$), these
edges can be combined with edges of the first type in order to obtain
the full component. This is why we said earlier that the truncated
component is the ``important'' part of the component.

The paths from $\vrtx_0$ to $\vrtx_{*}$ in $G$ correspond to the ways
in which we can concatenate path components together to get a full
path through the grid; if we assign the appropriate weights to the
edges, then the cost of a path through $G$ corresponds to the cost of
the same path through the grid. The distance from $\vrtx_0$ to
$\vrtx_{*}$ is therefore a $(1 + \Oreg(\eps))$-approximation for
$\DTW(x, y)$.

Finally, we must bound the size of $ G $. Each block
contains at most four intersection points; so there are $\Oreg(k\ell)$
total intersection points. Each intersection point creates at most
$\Oreg(\log_{1+\eps}(n))$ snap points; so there are $\Oreg(k \ell \eps^{-1}
\log n)$ snap points (which are the vertices in $V$). Each snap point
has an out-degree of at most $\poly(\eps^{-1} \log n)$. Hence we have:
\[
|E| \le \Oreg(k \ell \eps^{-1} \log n) \poly(\eps^{-1} \log n) =  \tilde{O}(k \ell / \poly(\eps)).
\]
We can therefore compute the distance from $\vrtx_0$ to $\vrtx_{*}$ in
$\tilde{O}(k \ell / \poly(\eps))$ time, as desired.

\paragraph{Paper outline.}
For the sake of simplicity, there are a number of details that we
chose to ignore in this section (such as a time-efficient construction
of $G$ and a careful proof that the modifications to $P$ incur only a
$(1+\Oreg(\eps))$-factor change in its cost). In the remainder of the
paper, we give a formal presentation and analysis of the algorithm
outlined above.

\section
{Preliminaries}
\label{section:preliminaries}

\def\myrun#1#2{{#1}{\char094}{#2}}

\begin%
{figure}[t]
\centering
\begin%
{tikzpicture}[scale=0.75]
\draw[step=1cm,gray,thin](0,0) grid (10,6);
\filldraw[fill=gray](0,0) rectangle (1,1);
\filldraw[fill=gray](1,1) rectangle (2,2);
\filldraw[fill=gray](2,1) rectangle (3,2);
\filldraw[fill=gray](3,2) rectangle (4,3);
\filldraw[fill=gray](4,2) rectangle (5,3);
\filldraw[fill=gray](5,2) rectangle (6,3);
\filldraw[fill=gray](6,3) rectangle (7,4);
\filldraw[fill=gray](7,4) rectangle (8,5);
\filldraw[fill=gray](8,5) rectangle (9,6);
\filldraw[fill=gray](9,5) rectangle (10,6);
\foreach \y in {0.5,1.5}
\foreach \x in {0.5,1.5,2.5}
  \draw (\x cm,\y cm) -- (\x cm,\y cm) node {$0$};
\foreach \y in {0.5,1.5}
\foreach \x in {3.5,4.5,5.5,6.5}
  \draw (\x cm,\y cm) -- (\x cm,\y cm) node {$1$};
\foreach \y in {0.5,1.5}
\foreach \x in {7.5,8.5,9.5}
  \draw (\x cm,\y cm) -- (\x cm,\y cm) node {$3$};
\foreach \y in {2.5}
\foreach \x in {0.5,1.5,2.5}
  \draw (\x cm,\y cm) -- (\x cm,\y cm) node {$1$};
\foreach \y in {2.5}
\foreach \x in {3.5,4.5,5.5,6.5}
  \draw (\x cm,\y cm) -- (\x cm,\y cm) node {$0$};
\foreach \y in {2.5}
\foreach \x in {7.5,8.5,9.5}
  \draw (\x cm,\y cm) -- (\x cm,\y cm) node {$2$};
\foreach \y in {3.5}
\foreach \x in {0.5,1.5,2.5}
  \draw (\x cm,\y cm) -- (\x cm,\y cm) node {$2$};
\foreach \y in {3.5}
\foreach \x in {3.5,4.5,5.5,6.5}
  \draw (\x cm,\y cm) -- (\x cm,\y cm) node {$1$};
\foreach \y in {3.5}
\foreach \x in {7.5,8.5,9.5}
  \draw (\x cm,\y cm) -- (\x cm,\y cm) node {$1$};
\foreach \y in {4.5,5.5}
\foreach \x in {0.5,1.5,2.5}
  \draw (\x cm,\y cm) -- (\x cm,\y cm) node {$3$};
\foreach \y in {4.5,5.5}
\foreach \x in {3.5,4.5,5.5,6.5}
  \draw (\x cm,\y cm) -- (\x cm,\y cm) node {$2$};
\foreach \y in {4.5,5.5}
\foreach \x in {7.5,8.5,9.5}
  \draw (\x cm,\y cm) -- (\x cm,\y cm) node {$0$};
\foreach \y in {0.5}
  \draw (1pt,\y cm) -- (-1pt,\y cm) node[anchor=east] {$\chr{a}$};
\foreach \y in {1.5}
  \draw (1pt,\y cm) -- (-1pt,\y cm) node[anchor=east] {$\chr{a}$};
\foreach \y in {2.5}
  \draw (1pt,\y cm) -- (-1pt,\y cm) node[anchor=east] {$\chr{b}$};
\foreach \y in {3.5}
  \draw (1pt,\y cm) -- (-1pt,\y cm) node[anchor=east] {$\chr{c}$};
\foreach \y in {4.5}
  \draw (1pt,\y cm) -- (-1pt,\y cm) node[anchor=east] {$\chr{d}$};
\foreach \y in {5.5}
  \draw (1pt,\y cm) -- (-1pt,\y cm) node[anchor=east] {$\chr{d}$};
\foreach \y in {6.5}
  \draw (1pt,\y cm) -- (-1pt,\y cm) node[anchor=east] {$(  y  )$};
\foreach \x in {0.5}
  \draw (\x cm,1pt) -- (\x cm,-1pt) node[anchor=north] {$\chr{a}$};
\foreach \x in {1.5}
  \draw (\x cm,1pt) -- (\x cm,-1pt) node[anchor=north] {$\chr{a}$};
\foreach \x in {2.5}
  \draw (\x cm,1pt) -- (\x cm,-1pt) node[anchor=north] {$\chr{a}$};
\foreach \x in {3.5}
  \draw (\x cm,1pt) -- (\x cm,-1pt) node[anchor=north] {$\chr{b}$};
\foreach \x in {4.5}
  \draw (\x cm,1pt) -- (\x cm,-1pt) node[anchor=north] {$\chr{b}$};
\foreach \x in {5.5}
  \draw (\x cm,1pt) -- (\x cm,-1pt) node[anchor=north] {$\chr{b}$};
\foreach \x in {6.5}
  \draw (\x cm,1pt) -- (\x cm,-1pt) node[anchor=north] {$\chr{b}$};
\foreach \x in {7.5}
  \draw (\x cm,1pt) -- (\x cm,-1pt) node[anchor=north] {$\chr{d}$};
\foreach \x in {8.5}
  \draw (\x cm,1pt) -- (\x cm,-1pt) node[anchor=north] {$\chr{d}$};
\foreach \x in {9.5}
  \draw (\x cm,1pt) -- (\x cm,-1pt) node[anchor=north] {$\chr{d}$};
\foreach \x in {10.5}
  \draw (\x cm,1pt) -- (\x cm,-1pt) node[anchor=north] {$(  x  )$};
\end{tikzpicture}
\caption{An optimal full path of the order $(10, 6)$ whose cost equals $1$}
\label{figure:DTWexample1}
\end{figure}

\def\myhlen#1{\lambda_h(#1)}
\def\myvlen#1{\lambda_v(#1)}
We use $\myran{n_1,n_2}$ for the set $\myset{n_1,n_1+1,\ldots,n_2}$
consisting of all the integers between $n_1$ and $n_2$, inclusive, and
use $\myran{n}$ as a shorthand for $\myran{1,n}$.  We use $T_{m,n}$
for a table consisting of $m$ columns and $n$ rows and $T_{m,n}[i,j]$
for the entry on the $i$-th column and $j$-th row, where
$(i,j)\in\myran{m}\times\myran{n}$ is assumed. We may use $T$ for
$T_{m,n}$ if $m$ and $n$ can be readily inferred from the
context. Please note that an entry $T_{m, n}[i, j]$ in a table should
be distinguished from the value stored in the entry---when discussing
the value, we shall refer to it as the content of the entry $T_{m,n}[i, j]$.

\noindent{\bf Letters.}~\kern6pt
Let us assume an alphabet $\letters$, which is possibly infinite. We
use $\mydist$ for a distance function on letters such that
$\delta(a,a)=0$ for any $a \in \letters$. We do not require that
$\delta$ be symmetric or the triangular inequality
$\delta(a,c)\leq\delta(a,b)+\delta(b,c)$ hold for $\delta$.

\noindent{\bf Strings.}~\kern6pt
We use $x$ and $y$ for strings.  We write $x=(a_1,\ldots,a_m)$ for a
string consisting of $m$ letters such that $x[i]$ (often written as
$x_i$), the $i$-th letter in $x$, is $a_i$ for each $i\in\myran{m}$.
We use $\myrun{a}{n}$ for a string of $n$ occurrences of $a$, which is
also referred to as a {\it run} of a, and $\myrle{x}$ for a run-length
encoded (RLE) string, which consists of a sequence of runs.  We use
$\mylen{\myrle{x}}$ and $\myrlen{\myrle{x}}$ for the length and
r-length of $\myrle{x}$, which are $m_1+\cdots+m_k$ and $k$,
respectively, in the case
$\myrle{x}=(\myrun{a'_1}{m_1},\ldots,\myrun{a'_k}{m_k})$.

A run in a string $x$ is maximal if it is not contained in a longer
run in $x$. There is a unique run-length encoding $\myrle{x}$ of $x$
that consists of only maximal runs in $x$, and this encoding
$\myrle{x}$ is referred to as the RLE representation of $x$. We
also use $\myrlen{x}$ for the number of maximal runs in $x$ (and
thus $\myrlen{x}=\myrlen{\myrle{x}}$).

We use $p$ for points, which are just integer pairs.
\begin%
{definition}
Given a point $p_1=(i_1,j_1)$, another point
$p_2=(i_2,j_2)$ is a successor of $p_1$ if (1) $i_2=i_1+1$ and
$j_2=j_1$, or (2) $i_2=i_1$ and $j_2=j_1+1$, or (3) $i_2=i_1+1$ and
$j_2=j_1+1$.
\end{definition}

\begin%
{definition}
Let $P=\pathseq{p_1,\ldots,p_R}$ be a sequence of points such that
$p_r\in\myran{m}\times\myran{n}$ holds for each $1\leq r\leq R$. We
call $P$ a path of order $({m},{n})$ if $p_{r+1}$ is a successor of
$p_{r}$ for each $1\leq r < R$ (and this $P$ is sometimes also called
a ``warping path''~\cite{DBLP:journals/corr/abs-1903-03003}). Also, we
refer to a path of length $2$ as a step that connects a point to one
of its successors.
\end{definition}
We use $\WP(m,n)$ for the set of paths of order $(m,n)$.  A path
$P_1\in\WP(m,n)$ is a subpath of another path $P_2\in\WP(m,n)$ if
$P_1$ is contained in $P_2$ (as a consecutive segment).

\begin%
{definition}
Let $P_1$ and $P_2$ be two non-empty paths.  We use $P_1\pequ P_2$ to
mean that $P_1$ and $P_2$ begin at the same point and end at the same
point.
\end{definition}

\begin%
{definition}
Let $P_1$ and $P_2$ be two paths such that the first point of $P_2$,
if it exists, is the successor of the last point of $P_1$, if it
exists.  We write $P_1\pjuxt{P_2}$ to mean the concatenation of $P_1$
and $P_2$ (as sequences of points) that forms a path containing both
$P_1$ and $P_2$ as its subpaths.  In the case where both $P_1$ and
$P_2$ are non-empty, there is a step in $P_1\pjuxt{P_2}$ connecting
$P_1$ and $P_2$ that consists of the last point in $P_1$ and the first
point in $P_2$.
\end{definition}
Also, we write $P_1\pjoin{P_2}$ to mean $P_1\pjuxt{P'_2}$ where the
last point of $P_1$ is assumed to be the first point of $P_2$ and
$P'_2$ is the tail of $P_2$, that is, $P'_2$ is obtained from removing
the first point in $P_2$.  In other words, $P_1\pjuxt{P_2}$ implies
that $P_1$ and $P_2$ share no point while $P_1\pjoin{P_2}$ implies
that $P_1$ and $P_2$ share one point, which is the last point of $P_1$
and the first point of $P_2$.

\begin%
{definition}
Let $x=(a_1,\ldots,a_m)$ and $y=(b_1,\ldots,b_n)$ be two
strings.  For each path $P\in\WP(m,n)$, there is a value
$\CC_{x,y}(P)=\Sigma_{r=1}^{R}\mydist(a_{i_r},b_{j_r})$, where $P$
equals $((i_1,j_1),\ldots,(i_R,j_R))$. This value is often referred to
as the cost of $P$. We may write $\CC(P)$ for $\CC_{x,y}(P)$ if it is
clear from the context what $x$ and $y$ should be.
\end{definition}

We call each $P\in\WP(m,n)$ a full path if $(i_1,j_1)=(1,1)$ and
$(i_R,j_R)=(m,n)$. The DTW distance between $x$ and $y$, denoted by
$\DTW(x, y)$, is formally defined as the minimum of $\CC_{x,y}(P)$,
where $P$ ranges over the set of full paths of order $(m, n)$. Also, a
full path $P$ is referred to as an optimal full path if
$\CC_{x,y}(P)=\DTW(x,y)$.  Given there are only finitely many paths of
order $(m, n)$, there must exist one full path that is optimal.
As an example, the shaded squares in Figure~\ref{figure:DTWexample1}
illustrate the following full path of the order $(10, 6)$:
\[
\pathseq{(1,1),(2,2),(3,2),(4,3),(5,3),(6,3),(7,4),(8,5),(9,6),(10,6)}
\]
where the number in each square is the assumed distance between the
two corresponding letters (computed here as the difference between
their positions in the alphabet).

\begin%
{definition}
Let $P=\pathseq{(i_1,j_1),\ldots,(i_R,j_R)}$.
\begin%
{enumerate}
\item
$P$ is a v-path if all the $i_r$ are the same for $1\leq r\leq R$.
\item  
$P$ is a h-path if all the $j_r$ are the same for $1\leq r\leq R$.
\item  
$P$ is a d-path if $i_{r+1}=i_r+1$ and $j_{r+1}=j_r+1$ for $1\leq r < R$.
\end{enumerate}
Please recall that a step is a path of length $2$.  If a step is a
h-path/v-path/d-path, respectively, then it is a h-step/v-step/d-step,
respectively.
\end{definition}

\begin%
{definition}
Let $x=(a_1,\ldots,a_m)$ and $y=(b_1,\ldots,b_n)$ be two strings.  We
use $\TDTW(x,y)$ for the table $T_{m,n}$ such that the content of
$T_{m,n}[i,j]$ is $\mydist(a_i,b_j)$ for each $i\in\myran{m}$ and
$j\in\myran{n}$.
\end{definition}
We may use $\TDTW$ for $\TDTW(x,y)$ if $x$ and $y$ can be readily
inferred from the context.  If a table $\TDTW$ can be readily inferred
from the context, we often associate a point $(i,j)$ with the entry
$\TDTW[i,j]$ and think of a path
$P=\pathseq{(i_1,j_1),\ldots,(i_R,j_R)}$ as the sequence of entries
$\TDTW[i_r,j_r]$ for $1\leq r\leq R$.  As an example, a full path of
the order $(10,6)$ is given in Figure~\ref{figure:DTWexample1}, where
the path is indicated with the 10 shaded entries.

Suppose that the $i$th run ($j$th) in $x$ ($y$) consists of the
letters in $x$ ($y$) from position $i_1$ ($j_1$) to position $i_2$
($j_2$), inclusive. Then there is a corresponding block $\BB_{i,j}$
consisting of all the entries $\TDTW[u,v]$ for $i_1\leq u\leq i_2$ and
$j_1\leq v\leq j_2$.
Finally, we introduce notation for discussing specific blocks:

\begin%
{definition}\label{define:beta-h-v}
If there exists a block $B$ to the right of $\BB_{i,j}$ such that
$\delta(B) < \delta(\BB_{i,j})$, we use $\beta_h(\BB_{i,j})$ for such
a $B$ that is the closest to $\BB_{i,j}$.  In other words,
$\beta_h(\BB_{i,j})$ is $\BB_{i',j}$ for the least $i'$ satisfying $i
< i'$ and $\delta(\BB_{i',j}) < \delta(\BB_{i,j})$.  Similarly, if
there exists a block $B$ above $\BB_{i,j}$ such that $\delta(B) <
\delta(\BB_{i,j})$, then $\beta_v(\BB_{i,j})$ is $\BB_{i,j'}$ for the
least $j'$ satisfying $j < j'$ and $\delta(\BB_{i,j'}) <
\delta(\BB_{i,j})$.
\end{definition}

\subsection
{Computing DTW Distance with Graphs}
\label{subsection:DTW_with_graphs}
It is well known~\cite{vintzyuk1968} that one can turn the problem of
computing $\DTW(x,y)$ for two given strings $x$ and $y$ into a problem
of finding the shortest distance between two given vertices in some graph,
as follows.

Let $x=(a_1,\ldots,a_m)$ and $y=(b_1,\ldots,b_n)$.  We can construct a
directed graph $G_0=\pairg{V_0,E_0}$ such that
\begin%
{enumerate}
\item
there is a vertex $\vrtx_{i,j}\in V_0$ for each pair
$(i,j)\in\myran{m}\times\myran{n}$, and
\item
there is a directed edge $\edge(\vrtx_{i_1,j_1},\vrtx_{i_2,j_2})$ of length
$\delta(a_{i_1},b_{j_1})$ connecting $\vrtx_{i_1,j_1}$ to
$\vrtx_{i_2,j_2}$ whenever $(i_2,j_2)$ is a successor of $(i_1,j_1)$.
\end{enumerate}
We use $\GDTW(x,y)$ for this graph $G_0$ and use $v$ and $e$ to range
over $V_0$ and $E_0$, respectively.  We may also refer to each vertex
$\vrtx_{i,j}\in V_0$ simply as point $(i,j)$ if there is no risk of
confusion.  Clearly, $\mylen{V_0}$, the size of $V_0$, is $mn$, and
$\mylen{E_0}$, the size of $E_0$, is bounded by $3mn$ (since each point
can have at most $3$ successors).

As every warping path is naturally mapped to a path in the graph
$\GDTW(x,y)$ and vice versa, we can use $P$ to range over both warping
paths in $\TDTW(x,y)$ and paths in $\GDTW(x,y)$ without risking
confusion. Given a (non-empty) warping path $P$ in $\TDTW(x,y)$, we
use $\LL(P)$ for the length of the corresponding path of $P$ in
$\GDTW(x,y)$, which equals the cost of $P$ minus the cost associated
with the last point in $P$. Therefore, finding the value of
$\DTW(x,y)$ is equivalent to finding the shortest distance from
$\vrtx_{1,1}$ to $\vrtx_{m,n}$, which can be done by running some
version of Dijkstra's shortest distance algorithm.  Alternatively,
since $\GDTW(x, y)$ is acyclic, one can use dynamic programming to
find the shortest distance, in which case the running time becomes
$\Oreg(mn)$. This yields the classic dynamic-programming solution for
computing $\DTW(x, y)$~\cite{vintzyuk1968}.

The basic strategy that we use in this paper to design a DTW
approximation algorithm can be outlined as follows. Let
$G_0=\pairg{V_0,E_0}$ be the graph $\GDTW(x,y)$ given above.  We try
to construct a graph $G=\pairg{V,E}$ such that $V\subseteq V_0$ holds
and the length of each edge $e$ in $E$ that connects a vertex $v_1$ to
another vertex $v_2$ equals the shortest distance from $v_1$ to $v_2$
as is defined in $G_0$. Let $\DD_0$ and $\DD$ be the shortest distance
functions on the graphs $G_0$ and $G$, respectively. We attempt to
prove that
\[
\DD_0(\vrtx_{1,1},\vrtx_{m,n})\leq\DD(\vrtx_{1,1}, \vrtx_{m,n})\leq\alpha\cdot{\DD_0(\vrtx_{1,1},\vrtx_{m,n})}
\]
for some approximation ratio $\alpha>1$ (e.g., $\alpha=1+\eps$ for
$\eps>0$).  By running a shortest-path algorithm on $G$, we are able
to compute $\DD(\vrtx_{1,1}, \vrtx_{m,n})$ and thus obtain an
$\alpha$-approximation algorithm for $\DTW(x,y)$.  As the time
complexity of such an algorithm can be bounded by $\Oreg(\mylen{E})$
plus the time needed for constructing $G$, the key to finding a fast
algorithm is try to minimize $\mylen{E}$, the size of $E$ (while
ensuring that the construction of $G$ can be done in
$\Oreg(\mylen{E})$ time).

\section%
{A (1+$\eps$)-Approximation Algorithm for DTW}
\label{section:ApproxDTW2}

\def\nbeta{\beta^{*}}


In this section, we present and analyze a $(1
+ \epsilon)$-approximation algorithm for approximating the DTW
distance between two run-length encoded strings in near-quadratic
time.

Let $x=(a_1,\ldots,a_m)$ and $y=(b_1,\ldots,b_n)$ be two non-empty
strings. Following Section~\ref{subsection:DTW_with_graphs}, our
approach will be to construct a graph $G=\pairg{V,E}$ based on
$\GDTW(x,y)$, reducing the problem of computing a
$(1+\eps)$-approximation of $\DTW(x,y)$ to finding the shortest
distance between the vertex $\vrtx_{1,1}$ and the vertex $\vrtx_{m,n}$
in $G$.

We begin by describing the notions of h-to-v paths and v-to-h
paths. The role that these will play in our algorithm is that we will
show how to decompose any full path $P$ into a concatenation $P_1 +
P_2 + \cdots$ of h-to-v and v-to-h paths.

\begin%
{definition}
Let $x$ and $y$ be two non-empty strings. A horizontal-to-vertical (h-to-v) path in $\TDTW(x,y)$ is a path that
connects a point on the lower boundary of a block $B_{i,j_1}$ to
another point on the right boundary of $B_{i,j_2}$.  An h-to-v
component in a full path is a maximal h-to-v path that is not
contained in any longer h-to-v path in the same full path.  A
vertical-to-horizontal (v-to-h) path can be defined similarly.
\end{definition}

Note that an h-to-v path matches characters from a \emph{single run}
in $x$ to characters from (possibly multiple) runs in $y$. We now make
the (standard) observation that, when we are comparing a single run to
characters to a multi-run string, DTW behaves in a very natural way:

\begin%
{observation}
\label{obs:run-path-dtw}
Let $x=(a_1,\ldots,a_m)$ and $y=(b_1,\ldots,b_n)$ be two non-empty
strings.  Assume that $x$ is a run of some letter $a_0$, that is,
$a_0=a_i$ for $1\leq i\leq m$.
\begin%
{enumerate}
\item
If $m\leq n$, then we have:
$\DTW(x, y)=\Sigma_{j=1}^{n}\delta(a_0,b_j)$.
This case corresponds to a path of the form $D\pjoin{U}$,
where $D$ consists of only d-steps and $U$ only v-steps.
\item
If $m\geq n$, then we have:
$\DTW(x,y)=\Sigma_{j=1}^{n}\delta(a_0,b_j)+(m-n)\cdot\delta(a_0,b_0)$,
where $b_0$
is some $b_j$ closest to $a_0$, that is, $\delta(a_0,b_0)$
equals the minimum of $\delta(a_0,b_j)$ for $1\leq j\leq n$.
This case corresponds to a path of the form $D_1\pjoin{H}\pjoin{D_2}$,
where $D_1$ and $D_2$ consist of only d-steps and $H$ only h-steps. It should be further noted that, in this case, $H$ can be assumed to travel along the lower boundary of some block, without loss of generality.
\end{enumerate}
\end{observation}

We say that a path connecting $p$ and $q$ is optimal if its cost is
the least among all the paths connecting $p$ and $q$.  By merging the
two cases in Observation~\ref{obs:run-path-dtw}, we can assume that each
optimal h-to-v path is of the form $D_1\pjoin{H}\pjoin{D_2}\pjoin{U}$,
where any of the four sub-components can vanish. From now on, we can
use a 5-tuple $(p_1, q_1, q_2, p_2, p_3)$ (which may also be written as
$p_1q_1q_2p_2p_3$) to refer to an h-to-v path, where $p_1q_1$ is $D_1$,
$q_1q_2$ is $H$, $q_2p_2$ is $D_2$, and $p_2p_3$ is $U$.  Similarly,
each optimal v-to-h path is of the form
$D_1\pjoin{U}\pjoin{D_2}\pjoin{H}$, and we use a corresponding 5-tuple
representation to refer to a v-to-h path as well.

\begin%
{figure}[t]
\centering
\begin%
{tikzpicture}[scale=1]
\node[label=left:{$p_1$}](p1) at (0.00,0.00) {};
\node[label=left:{$p_2$}](p2) at (4.00,3.00) {};
\node[label=left:{$p_3$}](p3) at (4.00,3.50) {};
\node[label=above:{$q_1$}](p3) at (1.00,1.00) {};
\node[label=above:{$q_2$}](p4) at (2.00,1.00) {};
\node[label=right:{$p'_3=\mynear(p_3)$}](p5) at (4.00,3.85) {};
\draw[thick](0.00,0.00)--(1.00,1.00)--(2.00,1.00)--(4.00,3.00)--(4.00,3.50);
\node[label=right:{$p'_1$}](q1) at (0.25,0.00) {};
\node[label=right:{$p'_2=\mynear(p_2)$}](q2) at (4.00,3.25) {};
\node[label=below:{$q'_1$}](q3) at (1.25,1.00) {};
\node[label=right:{$q'_2$}](q4) at (2.35,1.00) {};
\node[label=right:{$q'_3$}](q5) at (4.00,2.65) {};
\draw[dashed,thick](0.25,0.00)--(1.25,1.00)--(2.35,1.00)--(4.00,2.65)--(4.00,3.25)--(4.00,3.85);
\end{tikzpicture}
\caption{For illustrating h-to-v path approximation}
\label{figure:h-to-v-path-approximation}
\end{figure}
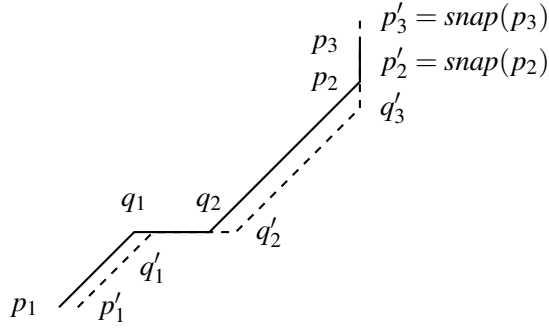

Next we argue that any full path $P_0$ can be decomposed into h-to-v and v-to-h paths.

\begin%
{lemma}
\label{lemma:run-path-seq}
Let $x$ and $y$ be two non-empty strings.  Given a full path $P_0$ in
$\TDTW(x,y)$, we have $P_0 = P_1 + \cdots + P_R$ where $P_r$ is an
h-to-v path for each odd $1\leq r\leq R$ and $P_r$ is a v-to-h path for each even
$1\leq r\leq R$.
\end{lemma}
\begin%
{proof}
We defer the proof to Appendix~\ref{appendix:ApproxDTW2}.
\end{proof}

Given two non-empty strings $x=(a_1,\ldots,a_m)$ and
$y=(b_1,\ldots,b_n)$, we outline as follows a strategy for
approximating $\DTW(x,y)$. Let $P_0$ be an optimal full path on
$\TDTW(x,y)$ such that $\CC(P_0)=\DTW(x,y)$.  By
Lemma~\ref{lemma:run-path-seq}, we have $P_0 = P_1 + \cdots + P_R$,
where $P_1$ is an h-to-v path and $P_1,\ldots,P_R$ are a sequence of
alternating h-to-v paths and v-to-h paths.  Let us choose an h-to-v
path $P_r$ for some $1\leq r\leq R$.  By
Observation~\ref{obs:run-path-dtw}, we can assume that $P_r$ is of the
form of solid lines depicted in
Figure~\ref{figure:h-to-v-path-approximation}.%
\footnote{The meaning of the dashed lines in the figure is to be explained later.}

In more detail, the path $P_r$ moves diagonally from a point $p_1$ on the lower
boundary of a block $B_1$ until it meets the lower boundary of another
block; it moves horizontally along that lower boundary for some distance; it then moves diagonally to reach a point $p_2$ on the right boundary of another block $B_2$ (which is
either $B_1$ or sits above $B_1$); and finally it moves vertically to reach a point $p_3$ on the right boundary of another block $B_3$ (which is
either $B_2$ or sits above $B_2$). Note that the horizontal moves
contained in $P_r$ must be inside a block where those moves cost the
least.\footnote{If there are several blocks in which such horizontal
moves can take place, we simply assume that the moves are inside the
lowest of these blocks.}

Let $G_0=\pairg{V_0,E_0}$ be the graph $\GDTW(x,y)$ described in
Section~\ref{subsection:DTW_with_graphs} for computing $\DTW(x,y)$.
We may use a point (that is, an integer pair) to refer to the
corresponding vertex in $V_0$. We may also use a path $P$ in
$\TDTW(x,y)$ to denote its counterpart in $G_0$. 

Given $\eps>0$, we will construct a graph $G=\pairg{V,E}$ such that:
\begin%
{itemize}
\item
The point $(1,1)$ is in $V$ and $V\subseteq V_0$ holds.
\item
Every point in $V$ is on a boundary. Each point in $V$ that is on either a right or upper boundary is connected by an edge to the next snap point on the same boundary (if there is one).
\item
If there is a step (either an h-step, v-step, or d-step) connecting two boundary points $p_1$ and $p_2$ in
$E_0$, then there is also a step connecting $\mynear(p_1)$ and
$\mynear(p_2)$ in $E$, where $\mynear(p_1)$ (resp. $\mynear(p_2)$) is
the nearest point in $V$ above or to the right of $p_1$ (resp. $p_2$) on the same boundary as $p_1$.
\item
For each h-to-v (resp. v-to-h) path $P$ from $p_1$ to $p_2$ (depicted
by some solid lines in Figure~\ref{figure:h-to-v-path-approximation})
in $G_0$ and any point $p'_1$ in $G$ to the right of $p_1$
(resp. above $p_1$) such that $p_1$ and $p'_1$ are on the same block
boundary, there exists a path $P'$ (depicted by some dashed lines in
Figure~\ref{figure:h-to-v-path-approximation}) in $G$ connecting
$p'_1$ and the point $p'_2=\mynear(p_2)$ in $G$ such that $\DD(P')\leq
(1+\eps)\DD_0(P)$, where $\DD$ and $\DD_0$ are the shortest distance
functions on the graphs $G$ and $G_0$, respectively.
\end{itemize}

We remark that our construction of $G$ will repeatedly make use of the
following basic fact.
\def\floor#1{\lfloor{#1}\rfloor}
\begin%
{observation}\label{obs:Delta}
Let
$\Delta(t)=\floor{(1+\eps)^{t}}$
for integers $t\geq 0$.
For each integer $d\geq 1$,
we have
a $(1+\eps)$-approximation of $d$
that is of the form $\Delta(t)$.
In other words,
$d\leq\Delta(t)\leq(1+\eps)\cdot{d}$ holds for some $t$.
\end{observation}

We now describe how to construct the graph $G$. We construct the set $V$ of vertices as follows:
\begin
{enumerate}
\item
Each vertex in $G_0$ corresponding to a corner point in
$\TDTW(x,y)$ should be added into $V$. There are at most $4k\ell$ such
vertices, where $k=\myrlen{x}$ and $\ell=\myrlen{y}$.
\item
Assume
$(i, j)$ is the lower-left corner of block $B$.
\begin%
{itemize}
\item
If a point
$(i+1+\Delta(t), j)$
is on the lower boundary of $B$
for some $t\geq 0$,
then this point should be added into $V$.
There are at most $\log_{1+\eps}(m)$ such points for the block $B$.
\item
If a point
$(i, j+1+\Delta(t))$
is on the left boundary of $B$
for some $t\geq 0$,
then this point should be added into $V$.
There are at most $\log_{1+\eps}(n)$ such points for the block $B$.
\end{itemize}
\item
Assume
$(i, j)$ is the upper-left corner of block $B$.
If a point
$(i+\Delta(t), j)$ is on the upper boundary of $B$
for some $t\geq 0$,
then this point should be added into $V$.
There are at most $\log_{1+\eps}(m)$ such points for the block $B$.
\item
Assume
$(i, j)$ is the lower-right corner of block $B$.
If a point
$(i, j+\Delta(t))$ is on the right boundary of $B$
for some $t\geq 0$,
then this point should be added into $V$.
There are at most $\log_{1+\eps}(n)$ such points for the block $B$.
\end{enumerate}
The points in $V$ are referred to as \defn{snap points}.  Given a snap
point $p'$ on the upper or right boundary of some block, if $p'q'$ is
a d-step for some point $q'$, then $q'$ is also a snap point. This can
be readily verified by inspecting the construction of $V$.  Also, for
each block $B$, there are at most $\Oreg(\log(m+n)/\eps)$ points added
to $V$. Therefore, $\mylen{V}$, the size of $V$, is
$\Oreg(k\ell\cdot\log(m+n)/\eps)$ or simply $\Otld(k\ell/\eps)$.

\begin%
{definition}
Given a point $p\in V_0$ on a horizontal
boundary of a block $B$, we use $\mynear_h(p)$
for the point $p'\in V$ such that $p'$ is $p$ if $p\in V$ or $p'$ is
the closest point to the right of $p$ that is on the same boundary of $B$.
The existence of such a point is guaranteed
as all of the corner points are included in $V$.
Similarly, $\mynear_v(p)$ can be defined for each point $p$ on a vertical boundary
of a block.
\end{definition}
We can use $\mynear(p)$ for either $\mynear_h(p)$ or $\mynear_v(p)$ without
confusion: If both $\mynear_h(p)$ and $\mynear_v(p)$ are defined for $p$,
then $p$ must be a corner of some block $B$, implying
$p=\mynear_h(p)=\mynear_v(p)$ since $p\in V$ holds.  We argue as follows
that finding $\mynear(p)$ for each given $p$ can be done in
$\Otld(1)$ time.\footnote{ We slightly abuse the $\Otld$ notation here
  as the parameters $m$ and $n$ for the implicit log-terms are not explicitly
  mentioned.}

\def\BB{\mbox{\bf B}}
\def\myhat#1{\hat{#1}}
\def\ix{\myhat{i}_{x}}
\def\jy{\myhat{j}_{y}}
\begin%
{definition}
\label{lemma:ij-hat}
Let $x=(a_1,\ldots,a_m)$ be a string and
$\myrle{x}=(\myrun{a'_1}{m_1},\ldots,\myrun{a'_k}{m_k})$ be its RLE
representation.  Let $M_{r}$ be $m_1+\cdots+m_{r}$ for each $0\leq r <
k$.  For each $1\leq i\leq m$, we use $\ix$ for the pair $(i_0,i_1)$
such that $i=M_{i_0}+i_1$ for $1\leq i_1\leq m_{i_0+1}$.
\end{definition}
We may use ${\myhat{i}}$ for $\ix$ if $x$ can be readily inferred from the
context. 

It is worth taking a moment to verify that we can compute $\ix$ efficiently. Assume that an array storing $M_{r}$ for $0\leq r < k$
is already built (in $\Oreg(k)$ time). Given $i\in\myran{m}$, we can
perform binary search on the array to find $i_0$ in
$\Oreg(\log(k))$ time such that $M_{i_0} < i \leq M_{i_0+1}$; we can
then compute $\ix$ as $(i_0,i-M_{i_0})$.


It is also worth verifying that we can compute $\mynear(p)$ in $\Otld(1)$ time. Given a point $p=(i,j)$ on a boundary of some block $B$ in
$\TDTW(x,y)$, we can compute $\ix=(i_0,i_1)$ in $\Oreg(\log(k))$ time.
Similarly, we can compute $\jy=(j_0,j_1)$ in $\Oreg(\log(\ell))$ time.
We can locate the block $B$ as $\BB_{i_0+1,j_0+1}$, and then find
$\mynear(p)$ in $\Oreg(1)$ time (assuming $\log_{1+\eps}(i_1)$ and
$\log_{1+\eps}(j_1)$ can be computed in $\Oreg(1)$ time). Therefore,
given $p$, we can compute $\mynear(p)$ in $\Otld(1)$ time.


Having established that we can compute $\myhat{i}$ and $\mynear(p)$
efficiently, we are nearly ready to describe the construction of the
edges $E$. Our final task before doing so is to establish a bit more
notation for how to talk about blocks.

Please recall that $\beta_h(\BB_{i,j})$ (resp. $\beta_v(\BB_{i,j})$
refers to the closest block $\BB_{i',j}$ (resp. $\BB_{i,j'}$) such
that $\delta(\BB_{i',j}) < \delta(\BB_{i,j})$
(resp. $\delta(\BB_{i,j'}) < \delta(\BB_{i,j})$) holds.  If there is
no such a block, $\beta_h(\BB_{i,j})$ (resp. $\beta_h(\BB_{i,j})$) is
undefined.

\begin{definition}
\label{definition:nbeta}
We refer to $\BB_{i_1,j}, \cdots, \BB_{i_S,j}$ as a $\beta_h$-sequence
if $\BB_{i_{s+1},j}=\beta_h(\BB_{i_{s},j})$ for $1\leq s < S$. Let
$\nbeta_h(x,y)$ be the length of a longest $\beta_h$-sequence.
Clearly, we have $\nbeta_h(x,y)\leq\myrlen{x}$. Similarly, we refer to
$\BB_{i,j_1}, \cdots, \BB_{i,j_S}$ as a $\beta_v$-sequence if
$\BB_{i,j_{s+1}}=\beta_v(\BB_{i,j_s})$ for $1\leq s < S$. Let
$\nbeta_v(x,y)$ be the length of a longest $\beta_v$-sequence.
Clearly, we have $\nbeta_v(x,y)\leq\myrlen{y}$. Let
$\nbeta(x,y)=\max(\nbeta_h(x,y),\nbeta_v(x,y))$.
\end{definition}

Observe that if
the underlying distance function $\delta$ on letters is from Hamming
space, then $\nbeta(x,y)\leq 2$ for any $x$ and $y$. Slightly more
generally, if $\delta$ is bounded by a constant, then $\nbeta(x,y)$ is
bounded by the same constant plus one. Later in the section, in the proof of
Theorem \ref{theorem:ApproxDTW2-poly-n}, we will also see an important (and
much more general) case where $\nbeta(x,y)$ is guaranteed to be
$\Otld(1)$.

\def\qq{\tilde{q}}
\def\qq{\overline{q}}
\def\mysnaps{{\it snaps}}
We are now ready to explain the construction of the set $E$ of edges
for connecting vertices in $V$. The basic idea is to construct $E$ in
such a way that, for any constructed path $p'_1q'_1q'_2q_3'p'_2$ (as
is depicted in Figure~\ref{figure:h-to-v-path-approximation}), there
should be a path in $G$ going from $p'_1$ to $p'_2$ whose cost is at
most the cost of the path in $G_0$---this allows for the
graph $G$ to capture all such paths, and ultimately allows for $G$ to
be used in our approximation algorithm. Formally, the construction of
$E$ can be performed with the following steps:
\begin%
{enumerate}
\item
Note that
the corner points in $V_0$ are all in $V$.
The edges
connecting these corner points in $E_0$ should be added into $E$.
\item
Given a point $p'_1\in V$ on the upper boundary of a block $B$, if
there is a d-step from $p'_1$ to $p'_2$ (on the lower boundary of the
block above $B$), then $p'_2$ is in $V$ and an edge from $p'_1$ to
$p'_2$ should be added into $E$ whose length equals $\delta(B)$.

\item
Given a point $p'_1\in V$ on the right boundary of a block $B$, if
there is a d-step from $p'_1$ to $p'_2$ (on the left boundary of the
block to the right of $B$), then $p'_2$ is in $V$ and an edge from
$p'_1$ to $p'_2$ should be added into $E$ whose length equals
$\delta(B)$.

\item
If two points $p'_1=(i_1,j_1)$ and $p'_2=(i_2,j_2)$ in $V$ are on the
same horizontal or vertical boundary of a block $B$ such that $p'_2$
is the closest point above or to the right of $p'_1$, then an edge from
$p'_1$ to $p'_2$ should be added into $E$ whose length equals
$(i_2-i_1)\cdot\delta(B)$ (horizontal) or $(j_2-j_1)\cdot\delta(B)$ (vertical).

\item
Let $p'_1$ be a point in $V$ on the lower boundary of a block
$B_1$. This step adds into $E$ edges between $p'_1$ and certain chosen
snap points $q'_4$ such that there are h-to-v paths connecting $p'_1$
and $q'_4$.

Let us use $B_1^{1},\ldots,B_1^{S}$ for the sequence where
$B_1=B_1^{1}$ and $B_1^{s+1}=\beta_v(B_1^{s})$ for $1\leq s < S$ and
$\beta_v(B_1^{S})$ is undefined. Clearly, $S$ is bounded by
$\nbeta_v(x,y)$ (according to the definition of $\nbeta_v(x,y)$).  Let
$B'_1$ range over $B_1^{1},\ldots,B_1^{S}$.

Let $q'_1 = (i'_1, j'_1)$ be the point on the lower boundary of $B'_1$
such that the path connecting $p'_1$ and $q'_1$ consists of only
d-steps.  Let $\mysnaps_h(q'_1)$ be the set consisting of the point
$q'_1$, the points on the lower boundary of $B'_1$ of the form
$(i'_1+\Delta(t), j'_1)$ for some $t\geq 0$, and the lower-right
corner point of $B'_1$.  For each $q'_2$ ranging over the set
$\mysnaps_h(q'_1)$, there exists at most one point $q'_3$ on the right
boundary of some $B_2$ (which is either $B'_1$ or sits above $B'_1$)
such that the path connecting $q'_2$ and $q'_3$ consists of only
d-steps.  As this $q'_3$ may not be in $V$, we choose $q'_4$ to be
$\mynear_v(q'_3)$, which is in $V$ by definition.  Note that the path
$p'_1q'_1q'_2q'_3q'_4$ is an h-to-v path in $\TDTW(x,y)$.
We add into $E$ an edge between $p'_1$ and $q'_4$ for each $q'_4$.
The length of each added edge connecting $p'_1$ and $q'_4$ is the
shortest distance between $p'_1$ and $q'_4$, which, by
Lemma~\ref{lemma:h-to-v-path-cost}, can be computed in
$\Otld(1)$ time.

There is one $q'_1$ for each $B'_1$, and there are at most
$\log_{1+\eps}(m)$ many of $q'_2$ for each $q'_1$, and there is at
most one $q'_3$ for each $q'_2$ and one $q'_4$ for each
$q'_3$. Therefore, for each $p'_1$, there are at most
$\nbeta_v(x,y)\cdot\log_{1+\eps}(m)$ edges added into $E$.

\item
Let $p'_1$ be a point in $V$ on the left boundary of a block $B_1$.
This step adds into $E$ edges between $p'_1$ and certain chosen snap
points $q'_4$ such that there are v-to-h paths connecting $p'_1$ and
$q'_4$.  We omit the details that are parallel to those in the
previous step.  There are at most $\nbeta_h(x,y)\cdot\log_{1+\eps}(n)$
edges added into $E$ for each $p'_1$.
\end{enumerate}

Let us take a moment to discuss how to efficiently compute the lengths
of the edges added to $E$ during the construction of
$G=\pairg{V,E}$. That is, how to construct and determine the cost of
each dotted path $p'_1q'_1q'_2q'_3q'_4$ depicted in
Figure \ref{figure:h-to-v-path-approximation}.  (Note that
$q'_4=\mynear_v(q'_3)$ is not shown in the figure.)

\begin%
{lemma}
\label%
{lemma:beta-compute}
Let $x$ and $y$ be two non-empty strings.
For $k=\myrlen{x}$ and $\ell=\myrlen{y}$,
\begin%
{enumerate}
\item
we can compute $\beta_h(B)$ for all the blocks $B$ in $\TDTW(x,y)$
in $\Oreg(k\ell\cdot\log(k))$ time, and
\item
we can compute $\beta_v(B)$ for all the blocks $B$ in $\TDTW(x,y)$
in $\Oreg(k\ell\cdot\log(\ell))$ time.
\end{enumerate}\end{lemma}
\begin%
{proof}
Please see Appendix~\ref{appendix:ApproxDTW2} for details.
\end{proof}

\begin%
{lemma}
\label{lemma:h-to-v-path-cost}
For each h-to-v path $(p'_1,q'_1,q'_2,q'_3,q'_4)$, its length can be
computed in $\Otld(1)$ time if the five snap points $p'_1$, $q'_1$, $q'_2$,
$q'_3$, and $q'_4$ are given.
\end{lemma}
\begin%
{proof}
Please see Appendix~\ref{appendix:ApproxDTW2} for details.
\end{proof}
For brevity, we omit the obvious lemma parallel to
Lemma~\ref{lemma:h-to-v-path-cost} that is instead on computing
the lengths of v-to-h paths in $\Otld(1)$ time.


We are now in a position to state and prove the main theorems of the
paper. As noted earlier, the basic idea behind our $(1
+ \epsilon)$-approximation algorithms is to compute a path-distance
through the graph $G = (V, E)$, and show that this distance closely
approximates $\DTW(x, y)$.

We begin by stating a theorem that parameterizes its running time by
$\nbeta(x, y)$---we will then apply this result to obtain fast running
times in the cases where the distance function $\delta$ outputs either
$\Oreg(\log n)$-bit integer values
(Theorem \ref{theorem:ApproxDTW2-poly-n}) or
$\{0, 1\}$-values (Theorem \ref{theorem:ApproxDTW2-const}).
\begin%
{theorem}
\label{theorem:ApproxDTW2}
Let $x=(a_1,\ldots,a_m)$ and $y=(b_1,\ldots,b_n)$ be two non-empty
strings, and let $\myrle{x}$ and $\myrle{y}$ denote the run-length encoded versions of the two strings.
There exists a $(1+\eps)$-approximation algorithm (ApproxDTW) for
each $\eps>0$ that takes $\myrle{x}$ and $\myrle{y}$ as its input and
returns a value $\DTWtwo(x,y)$ satisfying $\DTW(x,y) \leq \DTWtwo(x,y)
\leq (1+\eps)\cdot\DTW(x,y)$. Moreover, the worst-case time complexity of this
algorithm is $\Otld(k\ell\cdot\nbeta(x,y)/\eps^2)$ for $k=\myrlen{x}$
and $\ell=\myrlen{y}$, where $\nbeta(x,y)$ is defined in
Definition~\ref{definition:nbeta}.
\end{theorem}
\begin%
{proof}
The analysis of the approximation ratio follows as in Section \ref{sec:technical}. For the full proof of the theorem, see Appendix~\ref{appendix:ApproxDTW2}.
\end{proof}


\subsection%
{Time-Bound for Polynomially-Bounded Letter Distances}
In this section, we present a variant of the algorithm ApproxDTW for
approximating $\DTW(x,y)$ under the general condition that the
distances between letters are integer values bounded by some
polynomial of the lengths of $x$ and $y$.  The time complexity of this
variant, which takes $\myrle{x}$ and $\myrle{y}$ as its input to
compute $\DTW(x,y)$, is $\Otld(k\ell/\eps^3)$ for $k=\myrlen{{x}}$
and $\ell=\myrlen{{y}}$.

\def\cpow{\mbox{cpow}}
\begin%
{definition}
Let $\delta$ be a distance function on letters such that
$\delta(a,b)\geq 1$ if $\delta(a,b)\neq 0$.  Given $\eps_1 > 0$, we
use $\delta_{\eps_1}$ for the distance function such that
$\delta_{\eps_1}(a,b)=0$ if $\delta(a,b)=0$, or $\delta_{\eps_1}(a,
b)=\cpow(1+\eps_1,\delta(a,b))$ if $\delta(a,b)\geq 1$, where
$\cpow(1+\eps_1,\alpha)$ equals $(1+\eps_1)^{t}$ for the least integer
$t$ such that $\alpha\leq(1+\eps_1)^{t}$ holds.
\end{definition}
Please note that
$\delta(a,b)\leq\delta_{\eps_1}(a,b)\leq(1+\eps_1)\cdot\delta(a,b)$ holds
for any letters $a$ and $b$.

\begin%
{lemma}
\label{lemma:delta_eps}
Let $\DTW(\delta)$ be the DTW distance function where the
underlying distance function for letters is $\delta$. Given $\eps_1>0$,
we have the following inequality for each pair of strings $x$ and $y$:
\[
\DTW(\delta_{\eps_1})(x,y)\leq(1+\eps_1)\cdot\DTW(\delta)(x,y)
\]
\end{lemma}
\begin%
{proof}
Let $P$ be an optimal full path such that its length based on $\delta$
equals $\DTW(\delta)(x,y)$.  We know that the length of $P$ based on
$\delta_{\eps_1}$ is bounded by $(1+\eps_1)\cdot\DTW(\delta)(x,y)$
since $\delta_{\eps_1}(a,b)\leq(1+\eps_1)\cdot\delta(a,b)$ holds for
any letters $a$ and $b$.  As $\DTW(\delta_{\eps_1})(x,y)$ is bounded
by the length of $P$ based on $\delta_{\eps_1}$, we have the claimed
inequality.
\end{proof}
Let $\eps_1>0$ and $\eps_2>0$.  By Lemma~\ref{lemma:delta_eps}, every
$(1+\eps_2)$-approximation algorithm for DTW based on
$\delta_{\eps_1}$ is a $(1+\eps_1)\cdot(1+\eps_2)$-approximation
algorithm for DTW based on $\delta$.  For each $\eps>0$, if we choose,
for example, $\eps_1=\eps/2-\eps^2/2$ and $\eps_2=\eps/2$, then we
have $(1+\eps_1)\cdot(1+\eps_2) < 1+\eps$, implying that every
$(1+\eps_2)$-approximation algorithm for DTW based on
$\delta_{\eps_1}$ is a $(1+\eps)$-approximation algorithm for DTW
based on $\delta$.

\begin%
{theorem}\label{theorem:ApproxDTW2-poly-n}
Let $x=(a_1,\ldots,a_m)$ and $y=(b_1,\ldots,b_n)$ be two non-empty
strings.

Assume that the underlying distance function $\delta$ satisfies
(1)~$\delta(a,b)\geq 1$ if $\delta(a,b)\neq 0$ and (2)~$\delta(a,b)$
is $\poly(m+n)$ for any letters $a$ in $x$ and $b$ in $y$.  There
exists a $(1+\eps)$-approximation algorithm for each $\eps > 0$ that
takes $\myrle{x}$ and $\myrle{y}$ as its input and returns a value $w$
satisfying $\DTW(x,y) \leq w\leq (1+\eps)\cdot\DTW(x,y)$.  And the
worst-case time complexity of this algorithm is $\Otld(k\ell/\eps^3)$ for
$k=\myrlen{{x}}$ and $\ell=\myrlen{{y}}$.
\end{theorem}
\begin%
{proof}
Let $\eps_1=\eps/2-\eps^2/2$ and $\eps_2=\eps/2$.  By
Theorem~\ref{theorem:ApproxDTW2}, ApproxDTW (as is
presented in the proof of Theorem~\ref{theorem:ApproxDTW2}) takes
$\myrle{x}$ and $\myrle{y}$ as input and returns a
$(1+\eps_2)$-approximation of $\DTW(\delta_{\eps_1})(x,y)$.  And the
time complexity of the algorithm is $\Otld(k\ell\cdot\nbeta(x,y)/{\eps_2}^2)$.

Note that there are only $\Oreg(\log_{1+\eps_1}(m+n))$-many distinct
values of $\delta_{\eps_1}(a,b)$ for $a$ and $b$ ranging over letters
in $x$ and $y$, respectively.  Hence, for the underlying distance
function $\delta_{\eps_1}$ on letters, $\nbeta_h(x,y)$ is
$\Oreg(\log_{1+\eps_1}(m+n))$ and $\nbeta_v(x,y)$ is also
$\Oreg(\log_{1+\eps_1}(m+n))$, which implies that $\nbeta(x,y)$ is
$\Oreg(\log_{1+\eps_1}(m+n))$ or simply $\Otld(1/\eps_1)$. Therefore,
we can use ApproxDTW to compute a $(1+\eps_2)$-approximation of
$\DTW(\delta_{\eps_1})(x,y)$ in
$\Otld(k\ell/{\eps_1}{\eps_2}^2)$ time.  Since any
$(1+\eps_2)$-approximation of $\DTW(\delta_{\eps_1})(x,y)$ is a
$(1+\eps)$-approximation of $\DTW(\delta)(x,y)$, we are done.
\end{proof}

\subsection%
{Time-Bound for Constant-Bounded Letter Distances}
Assume that there exists a constant $N$ such that $\delta(a,b)$ is an
integer less than $N$ for each pair $a$ and $b$ in $\letters$. For
instance, $(\Sigma,\delta)$ satisfies this condition if it is Hamming
space (for which $N$ can be set to $2$).
\begin%
{theorem}\label{theorem:ApproxDTW2-const}
Assume that $\delta(a,b)$ are $\Oreg(1)$ for $a,b\in\Sigma$.
Then ApproxDTW,
the $(1+\eps)$-approximation algorithm for DTW
given in the proof of Theorem~\ref{theorem:ApproxDTW2}, runs
in $\Otld(k\ell/\eps^2)$ time for $k=\myrlen{{x}}$ and
$\ell=\myrlen{{y}}$, where $\myrle{x}$ and $\myrle{y}$ are the
input of the algorithm.
\end{theorem}
\begin%
{proof}
This theorem follows from Theorem~\ref{theorem:ApproxDTW2} immediately
since $\nbeta(x,y)$ is $\Oreg(1)$.
\end{proof}

\section{Conclusion}
\label{section:conclusion}
We have presented in this paper an algorithm for approximating the
DTW distance between two RLE strings.  Trading accuracy for
efficiency, this algorithm is of (near) quadratic-time complexity
and thus, as can be expected, asymptotically faster than the exact DTW
algorithm of cubic-time
complexity~\cite{DBLP:journals/corr/abs-1903-03003}, which is
currently considered the state-of-art of its kind.


It will be interesting to further investigate whether there exist
asymptotically faster approximation algorithms for DTW than the one
presented in this paper.  In particular, it seems both interesting and
challenging to answer the open question as to whether there exists a
(near) quadratic-time algorithm for computing the (exact) DTW distance
between two RLE strings.

\section*{Acknowledgements}

The authors would like to thank Charles E. Leiserson for his helpful feedback and suggestions.

This research was funded by a Hertz Foundation Fellowship and an NSF GRFP Fellowhship. 
The research was also partially sponsored by the United States Air Force Research Laboratory and the United States Air Force Artificial Intelligence Accelerator and was accomplished under Cooperative Agreement Number FA8750-19-2-1000. The views and conclusions contained in this document are those of the authors and should not be interpreted as representing the official policies, either expressed or implied, of the United States Air Force or the U.S. Government. The U.S. Government is authorized to reproduce and distribute reprints for Government purposes notwithstanding any copyright notation herein.

\newpage
\appendix

\section%
{Omitted Proofs}
\label{appendix:ApproxDTW2}

\noindent
{\bf Lemma~\ref{lemma:run-path-seq}}~
Let $x$ and $y$ be two non-empty strings.  Given a full path $P_0$ in
$\TDTW(x,y)$, we have $P_0 = P_1 + \cdots + P_R$ where $P_r$ is an
h-to-v path for each odd $1\leq r\leq R$ and $P_r$ is a v-to-h path for each even
$1\leq r\leq R$.
\begin%
{proof}[Proof of Lemma~\ref{lemma:run-path-seq}]

For every point $p_1$ in $P_0$ that is on the lower boundary of a
block $B_{i,j_1}$, one can always find a point $p_2$ on the right
boundary of another block $B_{i,j_2}$ (as the last point in $P_0$ is
an upper-right corner point) such that the subpath in $P_0$ connecting
$p_1$ and $p_2$ is a maximal h-to-v path.  And the point immediately
following $p_2$ in $P_0$, if it exists, must be on the left boundary
of some block.

Similarly, for every point $p_1$ in $P_0$ that is on the left boundary
of a block $B_{i_1,j}$, one can always find a point $p_2$ on the upper
boundary of another block $B_{i_2,j}$ such that the subpath in $P_0$
connecting $p_1$ and $p_2$ is a maximal v-to-h path.  And the point
immediately following $p_2$ in $P_0$, if it exists, must be on the
lower boundary of some block.
\end{proof}


\begin%
{definition}
Let $x=(a_1,\ldots,a_m)$ and $y=(b_1,\ldots,b_n)$ be two non-empty
strings.  Let the RLE representations of $x$ and $y$ be
$\myrle{x}=(\myrun{a'_1}{m_1},\ldots,\myrun{a'_k}{m_k})$ and
$\myrle{y}=(\myrun{b'_1}{n_1},\ldots,\myrun{b'_\ell}{n_\ell})$, respectively.  Given a
block $\BB_{i,j}$ in $\TDTW(x,y)$, let us define $\mu_v(\BB_{i,j})$
and $\mu_h(\BB_{i,j})$ as follows:
\[
\begin%
{array}{llcccccc}
~~~~~~~~~~~~~ &
\mu_h(\BB_{i,j}) & = & m_1\cdot\delta(a'_1,b'_j)&+&\cdots&+&m_{i-1}\cdot\delta(a'_{i-1},b'_{j}) \\
~~~~~~~~~~~~~ &
\mu_v(\BB_{i,j}) & = & n_1\cdot\delta(a'_i,b'_1)&+&\cdots&+&n_{j-1}\cdot\delta(a'_{i},b'_{j-1}) \\
\end{array}\]
\end{definition}
It should be clear that $\mu_h(\BB_{i,j})$ and $\mu_v(\BB_{i,j})$
for $1\leq i\leq k$ and $1\leq j\leq l$ can be computed and then
stored in $\Oreg(k\ell)$ time.  Afterwards, we can resort to a simple
$\Oreg(1)$-time lookup to find each of these values.

\vspace{6pt}
\noindent
{\bf Lemma~\ref{lemma:beta-compute}}~
Let $x$ and $y$ be two non-empty strings and $k=\myrlen{x}$ and
$\ell=\myrlen{y}$.
\begin%
{enumerate}
\item
We can compute $\beta_h(B)$ for all the blocks in $\TDTW(x,y)$
in $\Oreg(k\ell\cdot\log(k))$ time.
\item
We can compute $\beta_v(B)$ for all the blocks in $\TDTW(x,y)$
in $\Oreg(k\ell\cdot\log(\ell))$ time.
\end{enumerate}
\begin%
{proof}[Proof of Lemma~\ref{lemma:beta-compute}]
Given $j_0$ satisfying $1\leq j_0\leq\ell$, we sketch a way to compute
$\beta_h(\BB_{i,j_0})$ for $1\leq i\leq k$ in
$\Oreg(\ell\cdot\log(k))$ time.  While traversing $\BB_{i,j_0}$ for
$1\leq i\leq k$, we use a max-heap to store those blocks $B$ for which
$\beta_h(B)$ are yet to be computed. The priority of a block $B$ is
simply $\delta(B$).  Let $\BB_{i,j_0}$ be the block being
currently processed; we remove from the max-heap all the blocks $B$
satisfying $\delta(\BB_{i,j_0}) < \delta(B)$ and define $\beta_h(B)$
to be $\BB_{i,j_0}$; we then insert $\BB_{i,j_0}$ into the max-heap
and proceed to process the block $\BB_{i+1,j_0}$ (if it exists). As both
inserting and removing a block from the max-heap takes
$\Oreg(\log(k))$ time, we can process all the blocks $\BB_{i,j_0}$ for
$1\leq i\leq k$ in $\Oreg(k\cdot\log(k))$ time.  Therefore, we can
compute $\beta_h(B)$ for all the blocks in $\TDTW(x,y)$ in
$\Oreg(k\ell\cdot\log(k))$ time. A similar argument can be made for
part~(2).
\end{proof}

\noindent
{\bf Lemma~\ref{lemma:h-to-v-path-cost}}~ For each h-to-v path
$p'_1q'_1q'_2q'_3q'_4$, its length can be computed in $\Otld(1)$ time
if the five points $p'_1$, $q'_1$, $q'_2$, $q'_3$, and $q'_4$ are given.
\begin
{proof}[Proof of Lemma~\ref{lemma:h-to-v-path-cost}] Let
$p'_1=(i_1,j_1)$ and $q'_4=(i_2,j_2)$.  Clearly, $i_{1,0}=i_{2,0}$ as
the path $p'_1q'_1q'_2q'_3q'_4$ is an h-to-v path.  Note that we can
compute ${\myhat{i}}_1=(i_{1,0},i_{1,1})$ in $\Otld(1)$ time.  Also,
we can compute ${\myhat{j}}_1=(j_{1,0},j_{1,1})$ and
${\myhat{j}}_2=(j_{2,0},j_{2,1})$ in $\Otld(1)$ time. Then the sum of
the lengths of the path $p'_1q'_1$ and the path $q'_2q'_3q'_4$ equals
the value $\mu_v(B_2)-\mu_v(B_1)+(j_{2,1}-1)\cdot\delta(B_2)$, where
$B_1$ and $B_2$ are $\BB_{i_{1,0}+1,j_{1,0}+1}$ and
$\BB_{i_{2,0}+1,j_{2,0}+1}$, respectively.  Given $q'_1$ and $q'_2$,
the length of the path $q'_1q'_2$ can be computed in $\Otld(1)$
time. Therefore, the length of the path $p'_1q'_1q'_2q'_3q'_4$ can be
computed in $\Otld(1)$ time.
\end{proof}

\noindent
{\bf Theorem~\ref{theorem:ApproxDTW2}}~
Let $x=(a_1,\ldots,a_m)$ and $y=(b_1,\ldots,b_n)$ be two non-empty
strings.
There exists a $(1+\eps)$-approximation algorithm (ApproxDTW)
for each $\eps>0$ that takes $\myrle{x}$ and $\myrle{y}$ as its input
and returns a value $\DTWtwo(x,y)$ satisfying $\DTW(x,y) \leq \DTWtwo(x,y)
\leq (1+\eps)\cdot\DTW(x,y)$. And the worst-case time complexity of this algorithm
is $\Otld(k\ell\cdot\nbeta(x,y)/\eps^2)$ for $k=\myrlen{{x}}$ and $\ell=\myrlen{{y}}$,
where $\nbeta(x,y)$ is defined in Definition~\ref{definition:nbeta}.
\begin%
{proof}[Proof of Theorem~\ref{theorem:ApproxDTW2}]
Given the RLE representations $\myrle{x}$ and $\myrle{y}$ of $x$ and
$y$, respectively, we can build a graph $G=\pairg{V,E}$ in
$\Otld(k\ell\cdot\nbeta(x,y)/\eps^2)$ time by following the steps
outlined in Section~\ref{section:ApproxDTW2}. Let $\DTWtwo(x,y)$ be the
shortest distance between $\vrtx_{1,1}$ and $\vrtx_{m,n}$ (defined in
$G$) plus $\TDTW[m,n]$ (that is, $\delta(a_m,b_n)$).

Clearly, $\DTW(x,y)\leq\DTWtwo(x,y)$ holds as every path in $G$ can be
seen as a path in $\GDTW(x,y)$.

In order to show $\DTWtwo(x,y)\leq(1+\eps)\cdot\DTW(x,y)$, we first
verify that the length of an h-to-v path $p'_1q'_1q'_2q'_3p'_2$
depicted with dashed lines in
Figure~\ref{figure:h-to-v-path-approximation} is bounded by $(1+\eps)$
times the length of its counterpart $p_1q_1q_2p_2$ depicted in solid
lines.  Note that we hereby omit a completely parallel argument of
approximating a v-to-h path for brevity.

Let us analyze the diagram in
Figure~\ref{figure:h-to-v-path-approximation}.  Note that we have the
following:
\[
\begin%
{array}{rcl}
\LL(p_1q_1q_2p_2) & = &
\LL(p_1q_1)+\LL(q_1q_2)+\LL(q_2p_2) \\
\LL(p'_1q'_1q'_2q'_3p'_2) & = & 
\LL(p'_1q'_1)+\LL(q'_1q'_2)+\LL(q'_2q'_3p'_2) \\
\end{array}\]
Hence, it suffices to show both
$\LL(q'_1q'_2)\leq(1+\eps)\cdot\LL(q_1q_2)$ and the following:
\[
\LL(p'_1q'_1)+\LL(q'_2q'_3p'_2)\leq(1+\eps)\cdot(\LL(p_1q_1)+\LL(q_2p_2))
\]

Assume that $p_1$ and $p'_1$ are on the lower boundary of some block
$B_1$ and $p_2$ and $p'_2$ are on the right boundary of some block
$B_2$. Also, assume that $q_1$ is on the lower boundary of some block
$B'_1$.  Note that $B'_1$ is either $B_1$ or sits above $B_1$ and
$B_2$ is either $B'_1$ or sits above $B'_1$.

If $q'_1$ is to the right of $q_2$, then $q'_2$ is $q'_1$.  Otherwise,
$q'_2$ is the leftmost point on the lower boundary of $B'_1$ such that
(1) $q'_2$ is to the right of $q_2$ and (2) $q'_2$ is of the form
$(i'_1,j'_1+\Delta(t))$ for some $t\geq 0$, where $q'_1=(i'_1,
j'_1)$, or it is at the right end of the boundary. Either way, we have
$\LL(q'_1q'_2)\leq(1+\eps)\cdot\LL(q'_1q_2)$ (see
Observation~\ref{obs:Delta}), which implies
$\LL(q'_1q'_2)\leq(1+\eps)\cdot\LL(q_1q_2)$.

Let $p_4$ (which is not shown in the diagram) be $\mylr(B_1)$, the
lower-right corner of $B_1$. Then $\mylr(B_2)$, the lower corner of
$B_2$, is between $p_4$ and $p_2$.  We have
$\LL(p_4p'_2)\leq(1+\eps)\cdot\LL(p_4p_2)$ since $p'_2$ is
$\mynear_v(p_2)$, which implies $\LL(p_2p'_2)\leq\eps\cdot\LL(p_4p_2)$.
Also, $\LL(p_4p_2)\leq\LL(p_1q_1)+\LL(q_2p_2)$ holds in this case,
implying $\LL(p_2p'_2)\leq\eps\cdot(\LL(p_1q_1)+\LL(q_2p_2))$.
Therefore, the length of the path $p'_1q'_1q'_2q'_3p'_2$ in dashed
lines is bounded by $(1+\eps)$ times the length of the path
$p_1q_1q_2p_2$ in solid lines.

Let $P_0$ be an optimal full path in $G_0=\GDTW(x,y)$.  By
Lemma~\ref{lemma:run-path-seq}, we have $P_0 = P_1 + \cdots + P_R$
where $P_1$ is an h-to-v path and $P_1,\ldots,P_R$ are a sequence of
alternating h-to-v paths and v-to-h paths.

Assume that each $P_r$ for $1\leq r\leq R$ connects a point $p_{r}$ to
another point $q_{r}$. Let $p'_{1}$ be $(1,1)$. Then we can find a
path $P'_1$ in $G$ connecting $p'_{1}$ to $q'_1=\mynear(q_1)\in V$
such that $\DD(P'_1)\leq (1+\eps)\DD_0(P_1)$ holds, where we use $\DD$ and
$\DD_0$ for the shortest distance functions in $G$ and $G_0$,
respectively.  If $R\geq 2$, then there is a step from $q_{1}$ to
$p_{2}$, which implies a step from $q'_{1}$ to
$p'_{2}=\mynear(p_{2})\in V$; and we can then find a path $P'_2$ in
$G$ connecting $p'_{2}$ to $q'_{2}=\mynear(q_{2})\in V$ such that
$\DD(P'_2)\leq (1+\eps)\DD_0(P_2)$ holds. Continuing with this
procedure, we can find paths $P'_1,\ldots,P'_R$ in $G$ such that
$\DD(P'_r)\leq (1+\eps)\DD_0(P_r)$ holds for each $1\leq r\leq R$.
Let $P'_0=P'_1+\cdots+P'_R$, and we have $\DD(P'_0)\leq
(1+\eps)\DD_0(P_0)$.  Therefore, the shortest distance from the point
$(1,1)$ to the point $(m,n)$ in $G$ is a $(1+\eps)$-approximation of
$\DTW(x,y)$ (since $P_0$ is a chosen optimal full path).

By inspecting the procedure given above for building $E$, we can
readily conclude that $\mylen{E}$, the size of $E$, is
$\Oreg((k\ell)\cdot\nbeta(x,y)\cdot\log^2_{1+\eps}(m+n))$. Also, computing
the shortest distance between two given points in $G=\pairg{V,E}$
can be done (with dynamic programming) in $\Oreg(\mylen{E})$ time or
simply $\Otld(k\ell\cdot\nbeta(x,y)/\eps^2)$ time.
\end{proof}

\end{document}